\documentclass[12pt]{article}

\usepackage[inline,shortlabels]{enumitem}

\usepackage{amssymb,amsmath,amsthm}
\usepackage[USenglish]{babel}
\usepackage{listings}
\usepackage{tcolorbox}
\usepackage{thmtools,thm-restate}
\usepackage{auxkling}     
\theoremstyle{plain}
\newtheorem{theorem}{Theorem}
\newtheorem{lemma}[theorem]{Lemma}

\newcommand*{\localsearch}{Randomized Local Search\xspace}
\newcommand*{\shortlocalsearch}{\texttt{RLS}\xspace}

\newcommand*{\dml}{Destructive Majorization Lemma\xspace}
\newcommand*{\shortdml}{DML\xspace}

\newcommand*{\avg}{\varnothing}
\newcommand*{\disc}{\operatorname{disc}}

\newcommand*{\eqd}{\stackrel{d}{{=}}}

\newcommand*{\leqd}{\preceq}

\newcommand*{\citet}[1]{\cite{#1}}
\newcommand*{\Citet}[1]{\cite{#1}}

\begin{document}

\title{Tight Load Balancing via Randomized Local Search\thanks{A preliminary version of this paper appeared in proceedings of 2017 IEEE International Parallel and Distributed Processing Symposium (IPDPS'17).}}

\author{Petra Berenbrink\thanks{Now affiliated with University of
        Hamburg, Germany. Email: \email{petra@sfu.ca}} 
        \ and Peter Kling\thanks{Now affiliated with University of
        Hamburg, Germany. Email: \email{pkling@sfu.ca}}\\
Simon Fraser University, Burnaby, Canada
\and Christopher Liaw\thanks{Supported by an NSERC graduate
        scholarship. Email: \email{cvliaw@cs.ubc.ca}} \ and Abbas Mehrabian\thanks{Affiliated with both UBC and SFU when this work was done. Supported by a PIMS Postdoctoral        Fellowship and an NSERC Postdoctoral Fellowship. Email:        \email{abbasmehrabian@gmail.com}.}\\
University of British Columbia, Vancouver, Canada}

\maketitle

\vspace{-1cm}
\begin{abstract}
\boldmath
We consider the following balls-into-bins process with $n$ bins and $m$ balls:
Each ball is equipped with a mutually independent exponential clock of rate 1.
Whenever a ball's clock rings, the ball samples a random bin and moves there if
the number of balls in the sampled bin is smaller than in its current bin.

This simple process models a typical load balancing problem where users (balls)
seek a selfish improvement of their assignment to resources (bins). From a game
theoretic perspective, this is a randomized approach to the well-known
\emph{KP-model}~\cite{Koutsoupias:1999aa}, while it is known as \localsearch
(\shortlocalsearch) in load balancing literature~\cite{goldberg,systemspaper}.
Up to now, the best bound on the expected time to reach perfect balance was
$\LDAUOmicron{{(\ln n)}^2+\ln(n)\cdot n^2/m}$ due to \citet{systemspaper}. We
improve this to an asymptotically tight $\LDAUOmicron{\ln(n)+n^2/m}$. Our
analysis is based on the crucial observation that performing \emph{destructive
moves} (reversals of \shortlocalsearch moves) cannot decrease the balancing
time. This allows us to simplify problem instances and to ignore
\enquote{inconvenient moves} in the analysis.
\unboldmath

\end{abstract}


\section{Introduction}\label{sec:introduction}

We consider a system of $n$ identical resources and $m$ identical users. Each
user must be assigned to exactly one resource. A user on resource $i$
experiences a load equal to the number of users on resource $i$. Users can
migrate between resources, and the goal is to find a fast and simple migration
strategy that reaches a \emph{perfectly balanced} state in which the load
experienced by the users differs by at most 1.

Such load balancing (or reallocation) problems are well-studied and have a
multitude of applications, ranging from scheduling in peer-to-peer
systems~\cite{Surana:2006aa} and channel allocation in wireless
networks~\cite{Petrova} to numerical applications such as computation of
dynamics~\cite{Boillat:1991aa}. In the past decade, scalability and robustness
concerns caused a shift from relatively complex centralized protocols to simple,
often randomized distributed protocols. Indeed, consider applications such as
multicore computers, routing in and between data centres, or load distribution
in peer-to-peer networks. Due to the sheer size of these systems,
maintainability issues, and increasing user demand, we have to push for
distributed protocols that are easy to implement and that do not rely on global
knowledge or user coordination.

\subsection{Protocol and Results in a Nutshell}

We analyze the following natural and simple load balancing process: Each of the
$m$ users is activated by an independent exponential clock of rate 1. Upon
activation, a user chooses one of the $n$ resources uniformly at random and
compares his currently experienced load with the load it would experience at the
new resource. He migrates to the new resource if and only if doing so does not
result in a worse load. From a game theoretic perspective, this is a simple
randomized approach to the well-known KP-model with unit weights and
capacities~\cite{Koutsoupias:1999aa}. In load balancing, this strategy is known
as \localsearch (\shortlocalsearch)~\cite{goldberg,systemspaper}. (Here, \enquote{local} refers to the closeness of two consecutive solutions
    in the solution space, since they differ by the placement of at most one
    ball.)

Our main result (\cref{thm:main}) is that \shortlocalsearch reaches perfect
balance in expected time $\LDAUOmicron{\ln(n)+ n^2/m}$, and with high
probability in time $\LDAUOmicron{\ln(n)+\ln(n) \cdot n^2/m}$. These bounds are
asymptotically tight and improve the previously best bounds~\cite{systemspaper}
by a logarithmic factor. At the heart of our improvement lies the simple but
tremendously useful \dml (\shortdml, see \cref{sec:results+overview}). The
\shortdml formalizes the intuition that performing \emph{destructive moves} (the
reverse of a move permitted by \shortlocalsearch) cannot result in a speedup.
This allows us to reverse unwanted protocol moves during the analysis.
Moreover, in any analysis phase we can simplify the given configuration to a
worst-case instance for that phase, reducing the number of cases to consider.

We continue with a survey of related literature in \cref{sec:related_work}.
Formal problem and protocol definitions are given in
\cref{sec:preliminaries}. \cref{sec:results+overview} states our main
result and introduces the above mentioned destructive moves argument. 
In \cref{app:aux} we gather some auxiliary results.
The
analysis of \shortlocalsearch can be found in
\cref{sec:analysis:homogeneous}. The paper closes with a short conclusion in
\cref{sec:conclusion}

\section{Related Work}\label{sec:related_work}

The following literature survey on load balancing adapts the balls-into-bins
terminology and the notation $n$ for the number of bins (resources) and $m$ for
the number of balls (users). (Note that $n$ and $m$ may be swapped in some other papers.)
Balls-into-bins games come in a vast number of variations and have a long
tradition to model load balancing and similar problems~\cite{Johnson:1977aa}.
Many of the most recent results consider the effect of the \enquote{power of 2
choices}~\cite{Mitzenmacher:2001aa} and processes that are in some form
\enquote{self-stabilizing}~\cite{Czumaj98,Berenbrink:2016aa,BCNPP15}. We refer
to~\cite{Berenbrink:2016aa,BCNPP15} for a recent and comprehensive overview of
these variants. Here, we focus on three classes of more closely related
balls-into-bins processes:
\begin{enumerate}[wide]
\item \emph{Local Search:}
    With respect to our work, the most relevant type of protocols are
    \emph{local search protocols}, where balls are sequentially activated and
    relocated with the goal to achieve a perfectly balanced situation. The term
    \enquote{local} is with respect to the solution space, since one step of the
    protocol changes the current solution by the placement of at most one ball.
    Protocols in this class are typically quite simple in that movement
    decisions depend only on the involved bins. Most closely related to our work
    are~\cite{goldberg,systemspaper}. They study exactly the same process
    (\shortlocalsearch). \Citet{goldberg} claims an $\LDAUOmicron{n^2}$ upper
    bound on the expected time to reach perfect balance. This is improved by
    \citet{systemspaper} to $\LDAUOmicron[small]{{(\ln(n))}^2+\ln(n)\cdot
    n^2/m}$, which is $\LDAUOmicron[small]{{(\ln (n))}^2}$ for the important
    case $m\gg n$. In the present work, we provide a tight upper bound of
    expected time $\LDAUTheta{\ln(n)+n^2/m}$ (i.e., $\LDAUTheta{\ln n}$ for
    $m\gg n$).

    \Citet{Czumaj:2003aa} study another local search protocol. Here, initially
    each ball picks two alternative bins and is placed arbitrarily in one of
    them. Now, in each step of the protocol a pair of bins $(b_1,b_2)$ is chosen
    uniformly at random. If there is a ball in $b_1$ with alternative bin $b_2$,
    then this ball is placed in the least loaded bin among $b_1$ and $b_2$. One
    of the major results in~\cite{Czumaj:2003aa} is that if the balls are
    initially placed via the power of 2 choices, then perfect balance is reached
    in $n^{\LDAUOmicron{1}}$ steps (the hidden constant is $\geq4$). In the same
    situation, \shortlocalsearch needs only $\LDAUOmicron{n^2}$ activations (see
    Phase~2 in \cref{sec:analysis:homogeneous}). Moreover, \shortlocalsearch
    can be started from an arbitrary load situation.

\item \emph{Selfish Load Balancing:}
    Another strain of related work has a game theoretical background and is
    published under the theme of \emph{selfish load
    balancing} (see~\cite{Vocking07} for a comprehensive survey). A key difference to local search protocols is that balls act
    simultaneously. This introduces a problem: moves that, in isolation, improve
    the load of a ball might become bad if many balls perform this move. So,
    while one can compare the balancing times to local search
    protocols, such a direct comparison should be taken with a grain of salt. (In one selfish load balancing time step all $m$ balls are activated.
            Similarly, in one time unit of \shortlocalsearch $m$ balls are activated
            in expectation.)
       In
    particular, the results below suggest that the time to perfect balance in
    selfish load balancing has an inherent dependency on $m$, while there is no
    such dependency for local search protocols.

    \Citet{evendar_2005} consider selfish load balancing protocols with global
    knowledge (e.g., the average load). This allows them to reach perfect
    balance in expected $\LDAUOmicron{\ln\ln m+\ln n}$ steps. \Citet{selfish}
    consider a protocol without global knowledge. Here, balls move to a randomly
    sampled bin with a probability depending on the load difference. They bound
    the expected balancing time by $\LDAUOmicron{\ln\ln m+n^4}$. In a follow-up
    work, \citet{selfishweighted} suggested another protocol with expected
    balancing time $\LDAUOmicron{\ln m+n\cdot\ln n}$. In comparison
    with~\cite{evendar_2005}, these results indicate that avoiding global
    knowledge in selfish load balancing might increase the dependency on $n$
    substantially.

\item \emph{Threshold Load Balancing:}
    A third series of articles evolves around the idea of \emph{threshold load
    balancing}: each ball has a threshold and moves with a certain probability
    to a random bin whenever its experienced load is above that
    threshold (note that~\cite{evendar_2005} also falls into this category). As in selfish load balancing, balls act simultaneously (resulting in a
    similar, seemingly inherent dependency on $m$). An interesting observation
    is that the \shortlocalsearch protocol can be seen as a (sequential)
    threshold protocol with an adaptive, local threshold (the sampled bin's
    load).

    \Citet{ackermann} introduced the idea of threshold load balancing and gave a
    protocol that balances up to a constant multiplicative factor in time
    $\LDAUOmicron{\ln m}$ and up to an additive constant in time
    $\LDAUOmicron{n^2\cdot\ln m}$.
    \Citet{threshold,DBLP:journals/corr/HoeferS13} extended these protocols to
    general graphs. Recent improvements by \citet{thresholdweighted} show that,
    on general graphs, one can balance up to a constant multiplicative factor in
    time $\LDAUOmicron{\tau_{\text{mix}}\cdot\ln m}$, $\tau_{\text{mix}}$ being
    the graph's mixing time.

\end{enumerate}

\section{Model and Notation}\label{sec:preliminaries}

We describe our load balancing problem in terms of balls and bins. There are $n$
bins (resources/processors) and $m$ balls (users/tasks). We use the shorthands
$\intcc{n}\coloneqq\set{1,2,\dots,n}$ and $\intcc{m}\coloneqq\set{1,2,\dots,m}$
for the set of bins and balls, respectively. A \emph{configuration}
$\bm{\ell}={(\ell_i)}_{i\in\intcc{n}}\in\N_0^n$ is an $n$-dimensional vector
with $\sum_{i\in\intcc{n}}\ell_i=m$. Its $i$-th component $\ell_i\in\N_0$
denotes the number of balls in bin $i$ (its \emph{load}).

We seek a simple distributed load balancing protocol (to be executed by each
ball) such that all bins end up with almost the same load. To define this
formally, let $\varnothing\coloneqq m/n$ denote the \emph{average load} of the
system. The \emph{discrepancy} of configuration $\bm{\ell}$ is
$\disc(\bm{\ell})\coloneqq\max_{i\in\intcc{n}}\abs{\ell_i-\avg}$. We say a
configuration $\bm{\ell}$ is \emph{$x$-balanced} if $\disc(\bm{\ell})\leq x$ and
\emph{perfectly balanced} if $\disc(\bm{\ell}) < 1$.

\paragraph{Protocol Description}
Let us formally describe the \emph{\localsearch} (\shortlocalsearch) protocol.
Each ball is equipped with an exponential clock of rate $1$, and the clocks are
mutually independent. A ball is \emph{activated} whenever its clock rings.
Consider a configuration $\bm{\ell}$ and assume ball $j\in\intcc{m}$ in bin
$i\in\intcc{n}$ is activated. Then, it chooses a \emph{destination bin}
$i'\in\intcc{n}$ uniformly at random and moves from $i$ to $i'$ if and only if
$\ell_{i}\geq \ell_{i'}+1$. Let us remark that the protocol studied by~\cite{goldberg,systemspaper} is slightly different:
    they allow movement from $i$ to $i'$ if $\ell_{i}> \ell_{i'}+1$. However,
    since the bins and the balls are identical, the two protocols have precisely
    the same balancing time.
\begin{tcolorbox} \textbf{\small\localsearch($\shortlocalsearch$)}
\begin{lstlisting}
{code executed by ball $j$ in bin $i$ when $j$ is activated}
sample random bin $i'$
if $\ell_{i}\geq \ell_{i'}+1$:
    move to bin $i'$
\end{lstlisting}
\mbox{}\vspace{-1.8em}
\end{tcolorbox}
This describes a continuous time stochastic process. Starting from an initial
configuration $\bm{\ell}$ we write $\bm{\ell}(t)$ for the configuration at time
$t$. Note that $\bm{\ell}$ depends on the random ball activations as well as the
random destination bin choices of each ball. Notice that \shortlocalsearch has
the desirable properties that the discrepancy never increases, the minimum load
never decreases, and the maximum load never increases.

\paragraph{Additional Notation}
We use $\N\coloneqq\set{1,2,\dots}$ for the (positive) natural numbers and
$\N_0\coloneqq\set{0,1,2\dots}$ to include zero in the natural numbers. For any
$k\in\N$ we define $H_k\coloneqq\sum_{i=1}^k1/i=\ln(k)+\LDAUOmicron{1}$ as the
$k$-th harmonic number. Given two random variables $A$ and $B$, we write
$B\preceq A$ if $A$ stochastically dominates $B$ (i.e., if $\Pr{A\geq
x}\geq\Pr{B\geq x}$ for all $x\in\R$), and we write $A\eqd B$ if $A$ and $B$ are
equal in distribution. We say an event $E$ holds with high probability (w.h.p.)
if $\Pr{E}\geq1-n^{-\LDAUOmega{1}}$. $\Bin(n,p)$ denotes a binomial random
variable with parameters $n$ and $p$, and $\Exp(\lambda)$ denotes an exponential
random variable with parameter $\lambda$.

\section{Results and Proof Outline}\label{sec:results+overview}

Our main result is the following theorem.
\begin{theorem}\label{thm:main}
Consider a system of $n$ identical bins and $m$ identical balls in an arbitrary
initial configuration. Let $T$ be the time when \shortlocalsearch reaches a
perfectly balanced configuration. We have $\Ex{T}=\LDAUOmicron{\ln(n)+n^2/m}$
and w.h.p.\@ $T=\LDAUOmicron{\ln(n)+\ln(n)\cdot n^2/m}$.
\end{theorem}

As observed in~\cite{systemspaper}, our bounds are asymptotically tight. Indeed,
assume that initially all balls are in the same bin. To reach a perfectly
balanced configuration, we need to activate at least $m-\varnothing$ balls. The
expected time to do so is at least
$\sum_{k=\avg+1}^m1/k=H_m-H_{\avg}=\LDAUOmega{\ln n}$, yielding the first term
in the lower bound. For the second term, suppose $\avg$ is an integer, and
consider a configuration in which exactly one bin has load $\avg+1$, one other
bin has load $\avg-1$, and every other bin has load $\avg$. The time to balance
perfectly is exactly the time until one of the balls in the overloaded bin is
activated and samples the underloaded bin. The latter is an exponential random
variable with parameter $(\avg+1)\cdot1/n$. Thus, the expected time to reach
perfect balance is $n/(\avg+1)=\LDAUOmega{n^2/m}$. Requiring a high probability
result gives an additional $\ln n$ factor.

\paragraph{Destructive Moves}
Before we continue, we state an auxiliary lemma which is used throughout our
analysis. In a nutshell, it states that reversing a ball movement of
\shortlocalsearch cannot improve the time to reach perfect balance. This
intuitively simple observation turns out to be extremely useful. Basically, we
will be able to reduce arbitrary initial configurations to \enquote{well shaped}
configurations (decreasing the number of cases to consider) and to ignore
certain (at the moment unwanted) moves of the protocol.

To formalize this idea, consider a configuration $\bm{\ell}$. We call a movement
of a ball from bin $i$ to bin $j$ \emph{destructive} if $\ell_i\leq \ell_j+1$.
Note that a movement is destructive if and only if it is the reversal of a valid
protocol move. Also note that, if $\ell_i = \ell_j + 1$, then a move from $i$ to
$j$ is both a valid protocol move and a destructive move (see
\cref{fig:move_types}). Such a move is called a \emph{neutral move}.

\begin{figure}
\centering
\includegraphics[width=\linewidth]{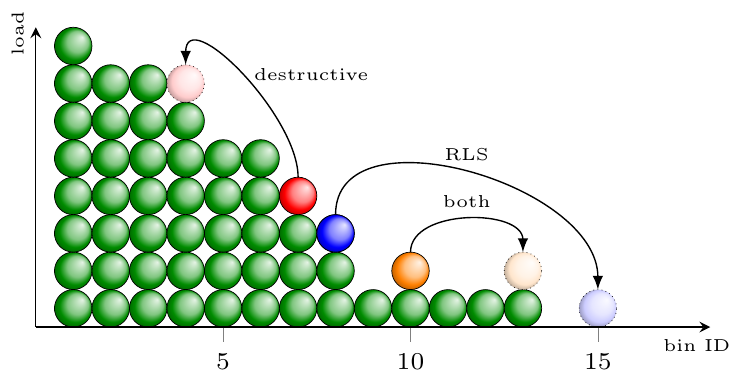}
\caption{Illustration of \shortlocalsearch moves versus destructive moves.}
\label{fig:move_types}
\end{figure}

\begin{restatable}[\dml]{lemma}{lemcoupling}\label{lem:coupling}
For any $t\geq0$, consider the load vector $\bm{\ell}(t)$ resulting from
protocol \shortlocalsearch at time $t$. Let $\bm{\tilde{\ell}}(t)$ denote the
load vector resulting from \shortlocalsearch at time $t$ under the presence of
an adversary who performs an arbitrary number of destructive moves after each ball
movement. Then $\disc(\bm{\ell}(t))\leqd\disc(\bm{\tilde{\ell}}(t))$.
\end{restatable}
We remark that
the adversary is allowed to have full knowledge of the protocol and its random choices
    (even future).

\begin{proof}
\newcommand*{\dstrSrc}{i_{\text{R}}}%
\newcommand*{\dstrDst}{i_{\text{L}}}%
\newcommand*{\moveSrc}{i_{\text{S}}}%
\newcommand*{\moveDst}{i_{\text{D}}}%
For $k\in\N_0$, consider the random process $P^{(k)}$ that executes our protocol
starting in the initial configuration $\bm{\ell}(0)$ under the presence of the
first $k$ adversarial (destructive) moves (ignoring any further adversarial
moves). Let $\bm{\ell^{(k)}}(t)$ denote the configuration at time $t$ under
process $P^{(k)}$. Note that $\bm{\ell^{(0)}}(t)\eqd\bm{\ell}(t)$ and
$\bm{\ell^{(\infty)}}(t)\eqd\bm{\tilde{\ell}}(t)$ for all $t \geq 0$. We show
that for all $k\in\N_0$ and $t\geq0$ we have
$\disc\bigl(\bm{\ell^{(k)}}(t)\bigr)\leqd\disc\bigl(\bm{\ell^{(k+1)}}(t)\bigr)$.
The lemma follows from this via the transitivity of stochastic domination.

Call a configuration $\bm{\ell'}$ \emph{close} to a configuration $\bm{\ell}$ if
$\bm{\ell'}$ is constructed from $\bm{\ell}$ by at most one destructive move.
Two immediate observations are:
\begin{enumerate}[(i)]
\item
    Either $\bm{\ell'}=\bm{\ell}$ or there are two bins $\dstrDst\neq\dstrSrc$
    with $\ell_{\dstrSrc}\leq \ell_{\dstrDst}+1$ such that
    $\ell'_{\dstrDst}=\ell_{\dstrDst}+1$, $\ell'_{\dstrSrc}=\ell_{\dstrSrc}-1$,
    and $\ell'_i=\ell_i$ for all $i\not\in\set{\dstrDst,\dstrSrc}$ (i.e.,
    $\bm{\ell'}$ is constructed from $\bm{\ell}$ by a destructive move from
    $\dstrSrc$ to $\dstrDst$. (If you think of the bins ordered non-increasingly, the destructive move
        goes from \textbf{R}ight ($\dstrSrc$) to \textbf{L}eft ($\dstrDst$).)
    
\item
    We have $\disc(\bm{\ell})\leq\disc(\bm{\ell'})$.
\end{enumerate}

Fix a $k\in\N_0$ and consider the processes $P^{(k)}$ and $P^{(k+1)}$.
Initially, we have $\bm{\ell^{(k)}}(0)=\bm{\ell^{(k+1)}}(0)$ and, thus,
$\bm{\ell^{(k+1)}}(0)$ is close to $\bm{\ell^{(k)}}(0)$. To show the
majorization
$\disc\bigl(\bm{\ell^{(k)}}(t)\bigr)\leqd\disc\bigl(\bm{\ell^{(k+1)}}(t)\bigr)$,
it is sufficient (by the above observations) to define a coupling between
$P^{(k)}$ and $P^{(k+1)}$ that maintains $\bm{\ell^{(k+1)}}(t)$ being close to
$\bm{\ell^{(k)}}(t)$ for all $t\geq0$. So fix $t\in\N$ and assume
$\bm{\ell^{(k+1)}}(t-1)$ is close to $\bm{\ell^{(k)}}(t-1)$. To simplify
notation, let $\bm{\ell}\coloneqq\bm{\ell^{(k)}}(t-1)$ and
$\bm{\ell'}\coloneqq\bm{\ell^{(k+1)}}(t-1)$. 

We claim that without loss of generality, we may let
both $\bm{\ell}$ and $\bm{\ell'}$ be sorted non-increasingly, such that $\ell_1
\geq \ell_2 \geq \dots \geq \ell_n$ and $\ell'_1 \geq \ell'_2 \geq \dots \geq
\ell'_n$.
The reason is that
    \shortlocalsearch is ignorant of the bin order, so we can assume it sorts
    configurations non-increasingly before each step. Also, sorting both
    $\bm{\ell}$ and $\bm{\ell'}$ maintains $\bm{\ell'}$ being close to
    $\bm{\ell}$: Assume $\bm{\ell}$ is sorted and construct $\bm{\ell'}$ from
    $\bm{\ell}$ by a destructive move from $\dstrSrc'$ to $\dstrDst'$. This
    implies $\dstrDst'\leq\dstrSrc'$. Let $\smash{\bm{\vec{\ell'}}}$ be the
    sorted version of $\bm{\ell'}$,
    $\dstrSrc\coloneqq\max\set{i\in\intcc{n}|\ell_i=\ell_{\dstrSrc'}}$, and
    $\dstrDst\coloneqq\min\set{i\in\intcc{n}|\ell_i=\ell_{\dstrDst'}}$. Then
    $\smash{\bm{\vec{\ell'}}}$ is constructed from $\bm{\ell}$ by a destructive
    move from $\dstrSrc$ to $\dstrDst$ and $\dstrDst\leq\dstrSrc$. 
    
    If $\bm{\ell'}=\bm{\ell}$ we use the identity coupling. Otherwise,
$\bm{\ell'}$ is constructed from $\bm{\ell}$ by a destructive move from a bin
$\dstrSrc$ to a bin $\dstrDst$ with $\dstrDst<\dstrSrc$. Without loss of
generality, let $m$ be the ball in which $\bm{\ell}$ and $\bm{\ell'}$ differ and
assume all other balls $j\in\intcc{m-1}$ are in the same bin in configuration
$\bm{\ell}$ and $\bm{\ell'}$. We couple the random choices of $P^{(k+1)}$ to the
random choices of $P^{(k)}$ as follows: Assume $P^{(k)}$ activates ball
$j\in\intcc{m}$ who is in source bin $\moveSrc$ in $\bm{\ell}$ and chooses
destination bin $\moveDst\in\intcc{n}$ (the $\moveDst$-th fullest bin in
$\bm{\ell}$). Then $P^{(k+1)}$ activates $j$ and chooses destination bin
$\moveDst$.

\begin{figure*}
\centering
\includegraphics{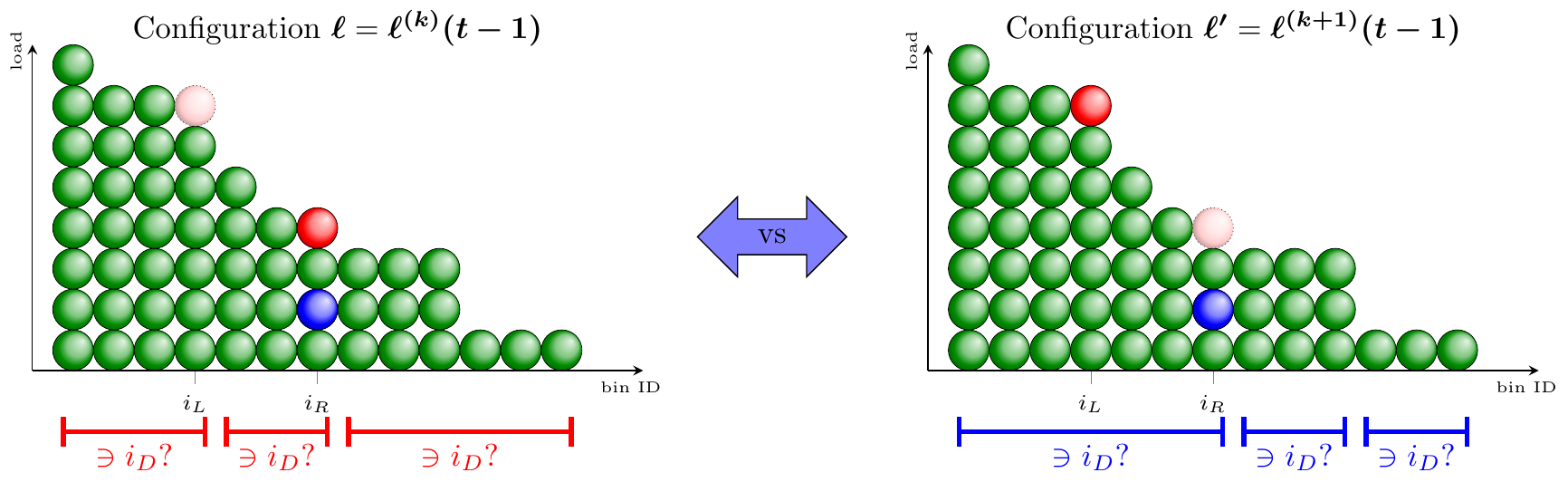}
\caption{%
    Illustration of the coupling used for \cref{lem:coupling}. The red ball
    is $m$. Thus, $\bm{\ell'}$ results from $\bm{\ell}$ by a destructive move
    from $\dstrSrc=7$ to $\dstrDst=4$. One of the balls (e.g., the blue or red
    ball) is activated in both processes and tries to move from its source to a
    random destination. If the red ball is activated, the source is $\moveSrc=7$
    on the left and $4$ on the right. If the blue ball is activated, the source
    is $\moveSrc=7$ both on the left and right. The destination $\moveDst$ is
    always the same on both sides. If the red ball is activated, the red
    intervals indicate the three subcases of the first part of
    Case~\ref{it:lem:coupling:case2} from the proof. If the blue ball is
    activated, the blue intervals indicate the three subcases of the second part
    of the same case (we drew the intervals below the right figure for space
    reasons). Similar pictures can be drawn for the remaining cases.
}
\label{fig:DML}
\end{figure*}
See \cref{fig:DML} for an illustration of the coupling and the following case
discrimination. It remains to show that the resulting configuration
$\bm{\ell^{(k+1)}}(t)$ of $P^{(k+1)}$ is close to the resulting configuration
$\bm{\ell^{(k)}}(t)$ of $P^{(k)}$. This is immediate if $\bm{\ell}=\bm{\ell'}$
(since then we use the identity coupling). Otherwise, $\bm{\ell'}$ is
constructed from $\bm{\ell}$ by a destructive move from a bin $\dstrSrc$ to a
bin $\dstrDst$ with $\dstrDst<\dstrSrc$. We distinguish the following cases
depending on the source and destination bins of process $P^{(k)}$:
\begin{enumerate}[(1)]
\item \underline{$\moveSrc,\moveDst \not\in \set{\dstrDst,\dstrSrc}$}
\label{destr:case:one}

    The processes behave identical and $\bm{\ell^{(k+1)}}(t)$ still results from
    $\bm{\ell^{(k)}}(t)$ by a destructive move from $\dstrSrc$ to $\dstrDst$.

\item \underline{$\moveSrc=\dstrSrc$}
\label{it:lem:coupling:case2}

    The activated ball might be $m$. If that is the case, $P^{(k)}$ activates a
    ball in bin $\dstrSrc$, while $P^{(k+1)}$ activates a ball in bin
    $\dstrDst$. We distinguish three subcases depending on the destination bin:
    If $\moveDst \leq \dstrDst$, both moves fail and nothing changes. If
    $\dstrDst < \moveDst \leq \dstrSrc$, only the move in $P^{(k+1)}$ succeeds
    and either the configurations become identical (if $\moveDst = \dstrSrc$) or
    $\bm{\ell^{(k+1)}}(t)$ results from $\bm{\ell^{(k)}}(t)$ by a destructive
    move from $\dstrSrc$ to $\moveDst$. If $\moveDst > \dstrSrc$, both moves
    succeed and the configurations become identical.

    If the activated ball is not $m$, both processes activate a ball in
    $\dstrSrc$. We distinguish three subcases depending on the destination bin:
    If $\moveDst \in \set{i\in\intcc{n}|\ell_i>\ell_{\dstrSrc}-1}$, both moves
    fail and nothing changes. If $\moveDst \in
    \set{i\in\intcc{n}|\ell_i=\ell_{\dstrSrc}-1}$, only the move in $P^{(k)}$
    succeeds (a neutral move) and $\bm{\ell^{(k+1)}}(t)$ results from
    $\bm{\ell^{(k)}}(t)$ by a destructive move from $\moveDst$ to $\dstrDst$.
    If $\moveDst \in \set{i\in\intcc{n}|\ell_i<\ell_{\dstrSrc}-1}$, both moves
    succeed and $\bm{\ell^{(k+1)}}(t)$ still results from $\bm{\ell^{(k)}}(t)$
    by a destructive move from $\dstrSrc$ to $\dstrDst$.

\item \underline{$\moveSrc=\dstrDst$}

    The activated ball cannot be $m$, so both processes activate a ball in bin
    $\dstrDst$. We distinguish three subcases depending on the destination bin:
    If $\moveDst \in \set{i\in\intcc{n}|\ell'_i>\ell'_{\dstrDst}-1}$, both moves
    fail and nothing changes. If $\moveDst \in
    \set{i\in\intcc{n}|\ell'_i=\ell'_{\dstrDst}-1}$, only the move in
    $P^{(k+1)}$ succeeds (a neutral move) and $\bm{\ell^{(k+1)}}(t)$ results
    from $\bm{\ell^{(k)}}(t)$ by a destructive move from $\dstrSrc$ to
    $\moveDst$. If $\moveDst \in
    \set{i\in\intcc{n}|\ell'_i<\ell'_{\dstrDst}-1}$, both moves succeed and
    $\bm{\ell^{(k+1)}}(t)$ still results from $\bm{\ell^{(k)}}(t)$ by a
    destructive move from $\dstrSrc$ to $\dstrDst$.

\item \underline{$\moveSrc \not\in \set{\dstrDst,\dstrSrc}, \moveDst=\dstrDst$}

    The activated ball cannot be $m$, so both processes activate a ball in the
    same bin $\moveSrc$. We distinguish three subcases depending on the source
    bin: If $\moveSrc \in \set{i\in\intcc{n}|\ell'_i<\ell'_{\dstrDst}}$, both
    moves fail and nothing changes. If $\moveSrc \in
    \set{i\in\intcc{n}|\ell'_i=\ell'_{\dstrDst}}$, only the move in $P^{(k)}$
    succeeds (a neutral move) and $\bm{\ell^{(k+1)}}(t)$ results from
    $\bm{\ell^{(k)}}(t)$ by a destructive move from $\dstrSrc$ to $\moveSrc$.
    If $\moveSrc \in \set{i\in\intcc{n}|\ell'_i>\ell'_{\dstrDst}}$, both moves
    succeed and $\bm{\ell^{(k+1)}}(t)$ still results from $\bm{\ell^{(k)}}(t)$
    by a destructive move from $\dstrSrc$ to $\dstrDst$.

\item \underline{$\moveSrc \not\in \set{\dstrDst,\dstrSrc}, \moveDst=\dstrSrc$}

    The activated ball cannot be $m$, so both processes activate a ball in the
    same bin $\moveSrc$. We distinguish three subcases depending on the source
    bin: If $\moveSrc \in \set{i\in\intcc{n}|\ell_i<\ell_{\dstrSrc}}$, both
    moves fail and nothing changes. If $\moveSrc \in
    \set{i\in\intcc{n}|\ell_i=\ell_{\dstrSrc}}$, only the move in $P^{(k+1)}$
    succeeds (a neutral move) and $\bm{\ell^{(k+1)}}(t)$ results from
    $\bm{\ell^{(k)}}(t)$ by a destructive move from $\moveSrc$ to $\dstrDst$.
    If $\moveSrc \in \set{i\in\intcc{n}|\ell_i>\ell_{\dstrSrc}}$, both moves
    succeed and $\bm{\ell^{(k+1)}}(t)$ still results from $\bm{\ell^{(k)}}(t)$
    by a destructive move from $\dstrSrc$ to $\dstrDst$.

\end{enumerate}
In all cases, $\bm{\ell^{(k+1)}}(t)$ is close to $\bm{\ell^{(k)}}(t)$. This
completes the proof.
\end{proof}

\section{Auxiliary Results}\label{app:aux}
\begin{lemma}[{Chernoff Bound, see~\cite[Theorem~4.4]{Mitzenmacher:2005aa}}]\label{lem:chernoff}
Consider the binomial distribution $\Bin(n,p)$ with parameters $n\in\N$ and $p\in\intcc{0,1}$.
Then, for any $\varepsilon\in\intcc{0,3/2}$ and $R\geq6np$ we have
\begin{align}
\Pr{\abs{\Bin(n,p)-np}>\varepsilon\cdot n p} &<    2e^{-\frac{\varepsilon^2\cdot n p}{3}}\quad\text{and}\label{chernoffsmalldev}\\
\Pr{\Bin(n,p)\geq R}                         &\leq 2^{-R}
.
\end{align}
\end{lemma}

\begin{lemma}[Concentration: Sum of Independent Exponentials]\label{lem:concentration_exponentials}
Let $X$ be a sum of independent exponential random variables, each having parameter $\geq\lambda$.
Then for any $\delta$
\begin{equation}
\Pr{X\geq\Ex{X}+\delta}\leq e^{\lambda^2\Var{X}/4-\lambda\delta/2}
.
\end{equation}
\end{lemma}
 \begin{proof}
 Let $X=\sum_{i=1}^{k}X_i$, with $X_i\eqd\Exp(\lambda_i)$ being independent exponential random variables with parameter $\lambda_i\geq\lambda$.
 Define $s\coloneqq\lambda/2$ and $\mu\coloneqq\Ex{X}=\sum_{i=1}^k\frac{1}{\lambda_i}$.
 By Markov's inequality, we have
 \begin{equation*}
 \begin{aligned}
        \Pr{X\geq\mu+\delta}
  =    \Pr{e^{sX}\geq e^{s\mu+s\delta}}
   \leq \frac{\Ex{e^{sX}}}{e^{s\mu+s\delta}}
  =    e^{-s\mu-s\delta}\cdot\prod_{i=1}^k\frac{\lambda_i}{\lambda_i-s}
 .
 \end{aligned}
 \end{equation*}
 Note that $s/\lambda_i\leq1/2$.
 Thus, using the inequality $1/(1-x)\leq e^{x+x^2}$ which holds for all $x\in\intcc{0,1/2}$, we have $\frac{\lambda_i}{\lambda_i-s}\leq e^{s/\lambda_i+s^2/\lambda_i^2}$.
 With this, we conclude
 \begin{equation*}
 \begin{aligned}
 &
 \Pr{X\geq\mu+\delta}
 \leq
 e^{-s\mu-s\delta}\cdot\prod_{i=1}^{k}e^{s/\lambda_i+s^2/\lambda_i^2} 
 \leq
 e^{-s\mu-s\delta+s\mu+s^2\Var{X}}
 =
 e^{\lambda^2\Var{X}/4-\lambda\delta/2}
 .\qedhere
 \end{aligned}
 \end{equation*}
 \end{proof}

\begin{lemma}[Concentration: Sum of Independent geometric random variables]\label{lem:concentration_geometric}
Let $Y_1,\dots,Y_k$ be independent geometric random variables with parameter $p\in\intco{0,1}$.
Define $L \coloneqq - \ln(1-p)$,
and let $c_1,\dots,c_k,M,S,V$ be positive constants satisfying
$M \coloneqq \max_i c_i$,
$S \geq \sum_i c_i$, and
$V \geq \sum_i c_i^2$.
Then for any $t$ we have
\begin{equation}
\Pr{\sum_i c_i Y_i \geq t}
\leq
\exp \left (
\frac{V}{4M^2}+\frac{S+SL-tL}{2M}
\right).
\end{equation}
\end{lemma}
 \begin{proof}
 For each $i$, define $Z_i \coloneqq Y_i - 1$.
 Let $X_1,X_2,\dots,X_k$ be independent exponential random variables such that $X_i$ has parameter $L/c_i$.
 Then,
 \begin{align*}
        \Pr{c_i Z_i \geq t}
 =    \Pr{Z_i \geq t/c_i}
   =    {(1-p)}^{\ceil{t/c_i}}
  \leq {(1-p)}^{t/c_i}
   =    e^{-Lt/c_i}
   =    \Pr{X_i \geq t}
 \end{align*}
 holds for all $i$ and $t$.
 Therefore,
 \begin{align*}
 \Pr{\sum_i c_i Y_i \geq t}
 & \leq
 \Pr{\sum_i c_i Z_i \geq t-S} \\
 & \leq
 \Pr{\sum_i X_i \geq t-S}.
 \end{align*}
 We will use \autoref{lem:concentration_exponentials} to bound the right hand side here.
 Note that all $X_i$ have parameters $\geq L/M \eqqcolon \lambda$.
 Moreover,
 $\Ex{\sum_i X_i} \leq S / L$
 and
 $\Var{\sum_i X_i} \leq V / L^2$.
 Therefore, \autoref{lem:concentration_exponentials} gives that
 \begin{align*}
         & \Pr{\sum_i X_i \geq t-S}\\
 {}\leq{}& \exp\left(\frac{\lambda^2 \Var{\sum_i X_i}}{4}-\frac{\lambda(t-S-\Ex{\sum_i X_i})}2\right)\\
 {}\leq{}& \exp\left(\frac{V}{4M^2}+\frac{S+SL-tL}{2M}\right)
 ,
 \end{align*}
 as required.
 \end{proof}

\begin{lemma}
\label{lem:exptowhp}
Let $ d_1 \leq d_2$ and suppose that, for any initial $d_2$-balanced configuration, the expected time to reach a $d_1$-balanced configuration is $t$.
Then, for any $d_2$-balanced configuration, the time to reach a $d_1$-balanced configuration is at most $2 t \log_2 n$ with high probability.
\end{lemma}

 \begin{proof}
 The crucial observation is that since $\bm{\ell}(0)$ is $d_2$-balanced, so is $\bm{\ell}(t)$ for any $t\geq0$.
 We partition the time interval $\intco{0,2 t \log_2 n}$
 into $\log_2 n$ epochs $\intco{0,2t}$, $\intco{2t,4t}$, and so on.
 We say the $i$-th epoch $\intco{2(i-1)t, 2i t}$ is \emph{successful} if $\disc(\bm{\ell}(2it)) \leq d_1$.
 Since we deterministically have $\disc(\bm{\ell}(2(i-1)t))\leq d_2$,
 regardless of the history up to time $2(i-1)t$,
 by Markov's inequality the probability that the $i$-th epoch is successful is at least $1/2$.
 Hence, the probability that none of the $\log_2 n$ epochs are successful is bounded by
 ${(1/2)}^{\log_2 n} =1/n$, as required.
 \end{proof}

\begin{lemma}
\label{lem:whptoexp}
Let $d_1 \leq d_2$ and suppose that, for any initial $d_2$-balanced configuration, the time to reach a $d_1$-balanced configuration is at most $t$ with probability at least $p$.
Let $Y$ be a geometric random variable with parameter $p$.
Then, for any $d_2$-balanced configuration, the time to reach a $d_1$-balanced configuration is stochastically dominated by $tY$, and so has expected value at most $t/p$.
\end{lemma}

 \begin{proof}
 The crucial observation is that since $\bm{\ell}(0)$ is $d_2$-balanced, so is $\bm{\ell}(t)$ for any $t\geq0$.
 We partition the time interval $\intco{0,\infty}$
 into epochs $\intco{0,t}$, $\intco{t,2t}$, and so on.
 We say the $i$-th epoch $\intco{(i-1)t, i t}$ is \emph{successful} if $\disc(\bm{\ell}(it)) \leq d_1$.
 Since we deterministically have $\disc(\bm{\ell}((i-1)t))\leq d_2$,
 regardless of the history up to time $(i-1)t$,
 the probability that the $i$-th epoch is successful is at least $p$.
 So, the index of the first successful epoch is dominated by a geometric random variable with parameter $p$, which has expected value $1/p$.
 \end{proof}

\section{Analysis of \shortlocalsearch}\label{sec:analysis:homogeneous}
In this section we analyze the \shortlocalsearch protocol.
For simplicity we would like to assume that $m \geq 2n$ and that $n$ divides $m$. The following two lemmas justify these assumptions.
\begin{lemma}
Suppose $m \leq n$ and let $T$ be the time until perfect balance.
Then $\Ex{T} \leq \LDAUOmicron{n}$ and $T \leq \LDAUOmicron{n \ln(n)}$ w.h.p.
\end{lemma}
\begin{proof}
By \cref{lem:coupling}, we may assume that all balls start in the first bin, and that we may wait for each of the $m$ balls to move to $m$ distinct empty bins (and ignore any other move).
Note that this is possible since $m \leq n$.
If there are $2\leq r\leq m$ balls left in the first bin,
then there are at least $r-1$ empty bins, hence
the time it takes for one of the $r$ balls to be activated and choose an empty bin is an exponential with rate $r \times (r-1)/n $.
Therefore, the total expected time for the $m$ balls to choose $m$ distinct empty bins is at most
\begin{equation}
  \sum_{r = 2}^m n/r(r-1)
< \sum_{r = 1}^{\infty} 2n/r^2 = \LDAUOmicron{n}.
\end{equation}
This shows $\Ex{T}=\LDAUOmicron{n}$.
\cref{lem:exptowhp} then implies $T \leq \LDAUOmicron{n \ln(n)}$ w.h.p.
\end{proof}
Note that the above lemma implies \cref{thm:main} when $m\leq n$.
The next lemma implies that, in order to prove \cref{thm:main}, it is sufficient to consider the case when $n$ divides $m$.
\begin{lemma}
Suppose $m \geq n$ and write $m = kn + r$ for some $k \in \N$ and $r\in\set{0,\dots,n-1}$.
Let $T$ be the time until perfect balance.
Furthermore, suppose that for any configuration with $n$ bins and $kn$ balls, \shortlocalsearch balances in expected time at most $f(n,kn)$ and in time at most $g(n,kn)$ w.h.p.
Then $\Ex{T} \leq \LDAUOmicron{\ln(n)} + f(n,kn)$ and $T \leq \LDAUOmicron{\ln(n)} + g(n,kn)$ w.h.p.
\end{lemma}
\begin{proof}
By \cref{lem:coupling}, we may assume that all balls start in the first bin.
We first wait for  $r$ balls in the first bin to move to $r$ distinct empty bins (and ignore any other move).
Once each of these balls finds a bin, we no longer allow it to move.
After these $r$ balls are moved to distinct bins, we run the \shortlocalsearch protocol assuming it had only $kn$ balls
(these assumptions can only slow down the protocol by \cref{lem:coupling}).
To complete the proof, we need only show that
the running time of the initial phase is $\LDAUOmicron{\ln n}$ in expectation and with high probability.
Denote this running time by $T'$.

Note that $T' = \sum_{i=1}^r T_i$, where $T_i$ is the time for the $i$th ball to activate and choose an empty bin.
Observe that $T_i$ is an exponential random variable
with parameter $(kn+r-i+1)(n-i)/n > n-i$.
Thus,
\begin{align}
      \Ex{T'}
&< \sum_{i=1}^r \frac{1}{n-i} \leq \LDAUOmicron{\ln(n)}\textnormal{, and}\\
      \Var{T'}
&< \sum_{i=1}^r \frac{1}{{(n-i)}^2} \leq \LDAUOmicron{1}
.
\end{align}
By concentration of sums of exponential random variables  (see \cref{lem:concentration_exponentials}),
we have $T'=\LDAUOmicron{\ln n}$ w.h.p.
\end{proof}

In light of the previous two lemmas, in the following we assume $m\geq n$ and that $n$ divides $m$.
Our analysis of \shortlocalsearch proceeds via the following three phases:
\begin{description}
\item[Phase 1:] In \cref{sec:analysis:homogeneous:phase1} we show that, from any initial configuration, w.h.p.\@ it takes time $\LDAUOmicron{\ln n}$ to become $\LDAUOmicron{\ln n}$-balanced.
\item[Phase 2:] In \cref{sec:analysis:homogeneous:phase2} we show that, from any $\LDAUOmicron{\ln n}$-balanced configuration, it takes expected time $\LDAUOmicron{n/\avg}$ to become $1$-balanced.
\item[Phase 3:] In \cref{sec:analysis:homogeneous:phase3} we show that, from any $1$-balanced configuration, it takes expected time $\LDAUOmicron{n/\avg}$ to become perfectly balanced.
\end{description}
Standard arguments imply that Phase 1 takes expected time $\LDAUOmicron{\ln n}$, and that Phases 2 and 3 take time $\LDAUOmicron{\ln n \cdot n/\avg}$ w.h.p.\@ (see \cref{lem:exptowhp} and \cref{lem:whptoexp} in the appendix).
Since $\avg=m/n$, these imply the total time to reach perfect balance is in expectation $\LDAUOmicron{\ln(n)+n^2/m}$ and w.h.p.\@ $\LDAUOmicron{\ln(n)+\ln(n) \cdot n^2/m}$.

\subsection{Phase 1: Reaching an \texorpdfstring{$\LDAUOmicron{\ln n}$}{O(ln(n))}-balanced Configuration}\label{sec:analysis:homogeneous:phase1}
We first consider how long it takes to go from an arbitrary initial configuration to a configuration $\bm{\ell}$ with $\disc(\bm{\ell})=\LDAUOmicron{\ln n}$.
We distinguish between two cases, depending on whether $\avg$ is large or small.

\paragraph{Phase 1 for small $\avg$}
The easier case, for $\avg\leq16\cdot\ln n$, is covered by the following lemma.
\begin{lemma}\label{lem:smallm}
Assume $\avg\leq16\cdot\ln n$ and consider an arbitrary initial configuration $\bm{\ell}=\bm{\ell}(0)$.
Let $T\coloneqq\inf\set{t|\disc(\bm{\ell}(t))\leq96\cdot\ln n}$.
Then, w.h.p., $T=\LDAUOmicron{\ln n}$.
\end{lemma}
\begin{proof}
By the assumption we have $\avg - 96\cdot\ln n\leq0$. Thus, $\ell_i(t) \geq 0
\geq \avg - 96\cdot\ln n$ for all bins $i\in\intcc{n}$ and times $t\geq0$.
Consequently, $T$ is the first time $t$ such that $\ell_{i}(t) \leq \avg +
96\cdot\ln n$ for all $i$. By \cref{lem:coupling}, we can assume that,
initially, all balls are in the same bin (by performing up to $m-1$ destructive
movements). So assume (w.l.o.g.) that all balls start in bin 1. Let us first
bound the time $T'$ until $m-\avg$ balls move from bin 1 to one of the other
$n-1$ bins. Applying \cref{lem:coupling} once more, we ignore movements of
balls to bin 1, movements between any of the remaining $n-1$ bins, and assume
that all movements from bin 1 to any of the remaining $n-1$ bins are successful.
If $T_i$ denotes the time in which the load of bin $1$ decreases from $i$ to
$i-1$, we have $T'=\sum_{i=\avg+1}^{m}T_i$. The different $T_i$ are independent
exponential random variables with parameter $i\cdot(n-1)/n$. This yields
\begin{align}
      \Ex{T'}
&=    \sum_{i=\avg+1}^m{\left(i\cdot\frac{n-1}{n}\right)}^{-1}
 \leq 2\cdot\ln n\textnormal{, and}\\
%
%
      \Var{T'}
&=    \sum_{i=\avg+1}^m{\left(i\cdot\frac{n-1}{n}\right)}^{-2}
 =    \LDAUOmicron{1/\avg}
.
\end{align}
By concentration of sums of independent exponential random variables (\cref{lem:concentration_exponentials} in Appendix~\ref{app:aux}), we have, w.h.p., $T'=\LDAUOmicron{\ln n}$.

To complete the proof, it suffices to show that w.h.p.\@ $T\leq T'$.
Whenever one of these $m-\avg$ balls is activated, we may assume it chooses one other bin uniformly at random and moves there (without checking the load).
While this might violate our original protocol (balls could move to a bin with a higher load), such a violation would be due to a destructive movement.
By \cref{lem:coupling}, this merely slows the process down.
Hence, at time $T'$, the number of balls in every other bin is $\Bin(m-\avg,1/(n-1))$, which has mean $(m-\avg)/(n-1)=\avg$.
Using a Chernoff bound (\cref{lem:chernoff}) with the union bound, the maximum of $n$ such binomials is not more than $96\cdot\ln n$ w.h.p.
\end{proof}

\paragraph{Phase 1 for large $\avg$}
We now turn to the more interesting case where $\avg>16\cdot\ln n$ and consider two subphases.
First, \cref{lem:first} shows that we reach a $\avg/2$-balanced configuration in time $\LDAUOmicron{\ln n}$.
The proof is basically identical to the proof of \cref{lem:smallm}, the only difference being that we use a different Chernoff bound at the end.
Afterward, \cref{lem:main} shows that it takes an additional $\LDAUOmicron{\ln n}$ time to become $\LDAUOmicron{\ln n}$-balanced.
\begin{lemma}\label{lem:first}
Assume $\avg>16\cdot\ln n$ and consider an arbitrary initial configuration $\bm{\ell}=\bm{\ell}(0)$.
Let $T\coloneqq\inf\set{t|\disc(\bm{\ell}(t))\leq\avg/2}$.
Then, w.h.p., $T=\LDAUOmicron{\ln n}$.
\end{lemma}
\begin{proof}[Proof Sketch]
The proof is basically identical to the proof of \autoref{lem:smallm}, the only difference being that we use a Chernoff bound for large expected values (Inequality~(\ref{chernoffsmalldev})) at the end.
As in the previous proof, by \autoref{lem:coupling} (destructive movements) we can assume all balls to be in bin 1 at time 0, and define $T'$ as the time until $m-\avg$ balls move to one of the other $n-1$ bins.
The same calculations yield (w.h.p.) $T'\leq \LDAUOmicron{\ln n}$.
Once more, \autoref{lem:coupling} allows us to majorize the ball distribution in each of these $n-1$ remaining bins at time $T'$ by the binomial distribution $\Bin(m-\avg,1/(n-1))$ with mean $\avg$.
Applying a Chernoff bound (this time the variant for large expected values, i.e., Equation~\eqref{chernoffsmalldev}) and a union bound yields that, w.h.p., the load of all bins is within $\intcc{\avg-2\sqrt{\avg\ln n},\avg+2\sqrt{\avg\ln n}}$.
Since $\avg>16\ln n$, this means w.h.p.\ $\disc(\bm{\ell}(T'))\leq2\sqrt{\avg\ln n}\leq\avg/2$.
\end{proof}

\begin{lemma}\label{lem:main}
Assume $m>n$ and consider an initial configuration $\bm{\ell}=\bm{\ell}(0)$ with $\disc(\bm{\ell})\leq\avg/2$.
Let $T\coloneqq\inf\set{t|\disc(\bm{\ell}(t))\leq 8 \ln n}$.
Then, w.h.p., $T \leq \LDAUOmicron{\ln n}$.
\end{lemma}
For proving \cref{lem:main}, we will apply the following lemma iteratively.
\begin{lemma}\label{lem:core}
Consider an initial configuration $\bm{\ell}=\bm{\ell}(0)$ with $\disc(\bm{\ell})\leq x$ for some $x\geq4\cdot\ln n$.
Let $T_x\coloneqq\inf\set{t|\disc(\bm{\ell}(t))\leq2\sqrt{x\cdot\ln n}}$.
Then, with probability $\geq 1- n^{-1}$ we have $T_x\leq\ln\bigl(\frac {\avg+x}{\avg-x}\bigr)$.
Moreover, $T_x$ is dominated by $Y\cdot\ln\bigl(\frac{\avg+x}{\avg-x}\bigr)$, where $Y$ is a geometric random variable with parameter $1-n^{-1}$.
\end{lemma}
\begin{proof}
First note that the second statement (domination by $Y\cdot\ln\bigl(\frac{\avg+x}{\avg-x}\bigr)$) follows from the first one via \cref{lem:whptoexp}.
We now prove the first statement.
Assume for simplicity that $\avg\pm x$ are integers.
Let $p\coloneqq2x/(x+\avg)$ and $t\coloneqq\ln(\avg+x)-\ln(\avg-x)$.
Note that the probability of activation of each ball during the interval $\intcc{0,t}$ is $1-\exp(-t)=p$.
Using \cref{lem:coupling} we make the following simplifying assumptions (see also \cref{fig:reordering}):
\begin{enumerate}
\item At time 0, we move some balls from the $n/2$ lightest bins (the \emph{light bins}) to the $n/2$ heaviest bins (the \emph{heavy bins}) in such a way that all light bins have exactly $\avg-x$ balls and all heavy bins have exactly $\avg+x$ balls.
    All these moves are destructive, thus we can assume (by \cref{lem:coupling}) that we start in the resulting configuration.
    Bins labeled as light/heavy in the beginning keep this label (regardless of how their loads changes) during the time interval $[0,t]$.
\item During the time interval $[0,t]$, we ignore activations of balls in light bins (as we could reverse them via \cref{lem:coupling}).
\item During the time interval $[0,t]$, we ignore movements between any two heavy bins (as we could reverse them via \cref{lem:coupling}).
\item During the time interval $[0,t]$, if a ball in a heavy bin $i$ is activated and tries to move to a light bin $i'$, it does so unconditionally (i.e., even if $\ell_{i} < \ell_{i'}+1$; in that case it is a destructive move which we may allow via \cref{lem:coupling}).
\end{enumerate}

\begin{figure*}
\centering
\includegraphics[width=\linewidth]{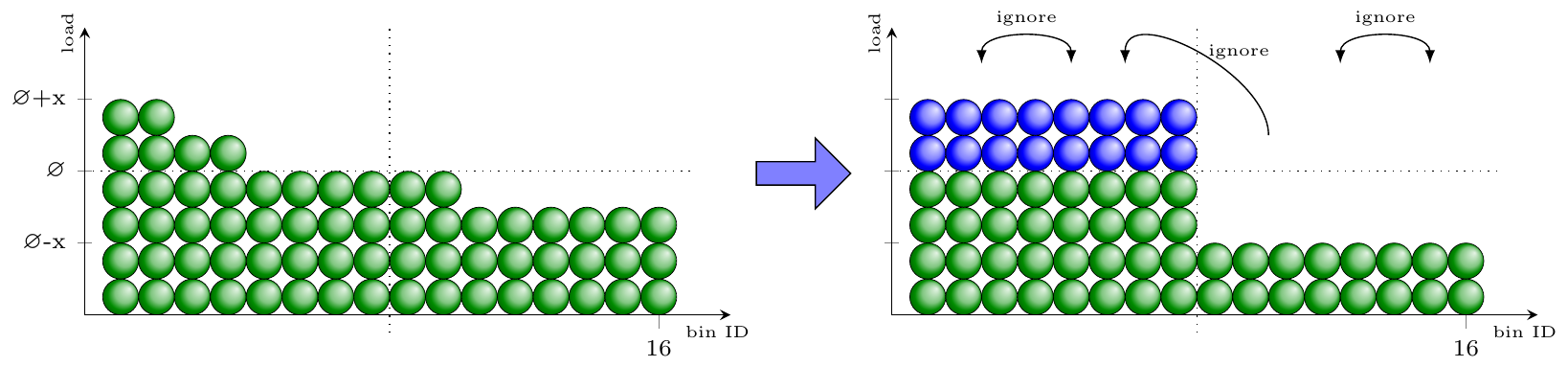}
\caption{%
    Using \cref{lem:coupling}, we reorder the balls such that the distance to the average is exactly $x=2$ in the first and last eight bins.
    We also allow only moves from heavy to light bins; all other moves are ignored.
}
\label{fig:reordering}
\end{figure*}

First, consider a heavy bin.
During the time interval $[0,t]$, each of its balls is activated with probability $p$, and moves to a light bin with probability $1/2$.
So this bin loses $\Bin(\avg+x,p/2)$ balls.
This binomial has expected value $x\geq4\ln n$.
Thus, by Chernoff, with probability $\geq1-n^{-2}$ its value is in $\intcc{x-2\sqrt{x\ln n},x+2\sqrt{x\ln n}}$.
This bin had $\avg+x$ balls initially, so with probability $\geq1-n^{-2}$, it will have between $\avg-2\sqrt{x\cdot\ln n}$ and $\avg+2\sqrt{x\cdot\ln n}$ balls at time $t$.

Next, consider a light bin.
There are $(\avg+x)(n/2)$ balls it can potentially receive during the time interval $[0,t]$.
It receives each one with probability $p/n$, so the number of balls it receives is $\Bin((\avg+x)(n/2),p/n)$.
This binomial has expected value $x\geq4\ln n$.
Thus, by Chernoff, with probability $\geq 1-n^{-2}$ its value is in $\intcc{x-2\sqrt{x\ln n},x+2\sqrt{x\ln n}}$.
This bin had $\avg-x$ balls initially, so with probability $\geq 1-n^{-2}$, it will have between $\avg-2\sqrt{x\cdot\ln n}$ and $\avg+2\sqrt{x\cdot\ln n}$ balls at time $t$.
Applying the union bound over all bins completes the lemma's proof.
\end{proof}
We now present the proof of \cref{lem:main}.
\begin{proof}
[Proof of \cref{lem:main}]
Note that, whenever $x \leq\avg/2$, we have
\(
\ln(\avg+x)-\ln(\avg-x)=\ln\bigl(1+\frac{2x}{\avg-x}\bigr)
\leq \frac{2x}{\avg-x}
\leq 4x/\avg
\).
Define $t_0\coloneqq0$, $x_0\coloneqq\avg/2$, and $x_k\coloneqq\sqrt{4x_{k-1}\cdot\ln n}$ for $k>0$.
By induction, $x_k\leq4\ln(n)\cdot\smash{x_0^{1/2^k}}$.
Let $r=\log_2\log_2\avg$.
Since $x_0\leq\avg$, we have $x_r\leq4\ln(n)\cdot\avg^{1/\log_2\avg}=8\cdot\ln n$.
Let $Y_1,Y_2,\dots$ be independent geometric random variables with parameter $1-n^{-1}$.
Applying \cref{lem:core} iteratively, the time to reach an $x_r$-balanced configuration is stochastically dominated by
\begin{equation}
     Z_r\coloneqq\sum_{i=0}^{r-1}Y_i\cdot4x_i/\avg
\leq \sum_{i=0}^{r-1}c_i\cdot Y_i
,
\end{equation}
where $c_i \coloneqq 16 \cdot \ln(n) \cdot \smash{x_0^{1/2^i}} / \avg$.
Straightforward calculations yield $\max c_i = \LDAUOmicron{\ln n}$, $\sum c_i =
\LDAUOmicron{\ln n}$, and  $\sum c_i^2 = \LDAUOmicron{\ln^2 n}$ (see below for the detailed calculations). By
concentration of sums of geometric random variables (see
\cref{lem:concentration_geometric}), we find that $Z_r = \LDAUOmicron{\ln n}$
w.h.p., completing the proof.

Detailed calculations: it only remains to show that
$\max c_i= \LDAUOmicron{\ln n}$,
$\sum c_i = \LDAUOmicron{\ln n}$,
and
$\sum c_i^2 = \LDAUOmicron{\ln^2 n}$.
First, we have $\max c_i=c_0=8\cdot\ln n$.  

Second, note that for any $y\geq1$ we have
$\sum_{i=0}^{r-1}y^{1/2^i}\leq2y+4r$.  Indeed, let $k$ be the smallest integer
such that $y^{1/2^k}<4$.  Then,
\begin{align*}
       \sum_{i=0}^{r-1}y^{1/2^i}
  =    \sum_{i=0}^{k}y^{1/2^i}+\sum_{i=k+1}^{r-1}y^{1/2^i}
  \leq \sum_{i=0}^{k}y/{2^i}+4r
  <    2y+4r.
\end{align*}
Using this, we bound
\begin{align*}
       \sum_{i=0}^{r-1}c_i
& =    \frac{16\cdot\ln(n)}{\avg}\cdot\sum_{i=0}^{r-1}x_0^{1/2^i}
  \leq \frac{16\cdot\ln(n)}{\avg}\cdot(2x_0+4r) \\
& =    \frac{16\cdot\ln(n)}{\avg}\cdot(\avg+4\log\log\avg)
  \leq 16\cdot\ln(n)\cdot2
\end{align*}

Finally, using a similar analysis we get
\begin{align*}
       \sum_{i=0}^{r-1}c_i^2
& =    {\left(\frac{16\cdot\ln(n)}{\avg}\right)}^2\cdot\sum_{i=0}^{r-1}x_0^{2/2^i}\\
& \leq {\left(\frac{16\cdot\ln(n)}{\avg}\right)}^2\cdot(2x_0^2+4r)\\
& =    {\left(\frac{16\cdot\ln(n)}{\avg}\right)}^2\cdot\left(\frac{\avg^2}{2}+4\cdot\log\log\avg\right) \\
& \leq {\left(\frac{16\cdot\ln(n)}{\avg}\right)}^2\cdot\avg^2=256\cdot{(\ln n)}^2.\qedhere
\end{align*}
\end{proof}

\subsection{Phase 2: Reaching a \texorpdfstring{$\bm{1}$}{1}-balanced Configuration}\label{sec:analysis:homogeneous:phase2}
\begin{lemma}\label{lem:potential}
Consider an initial configuration $\bm{\ell}=\bm{\ell}(0)$ with $\disc(\bm{\ell})=\LDAUOmicron{\ln n}$.
Let $T\coloneqq\inf\set{t|\disc(\bm{\ell}(t))\leq1}$.
Then $\Ex{T}=\LDAUOmicron{n/\avg}$.
\end{lemma}
Before we prove \cref{lem:potential}, let us introduce the notion of \emph{overloaded bins/balls}:
A bin $i$ is \emph{overloaded} if $\ell_i>\avg$.
The quantity
$\sum_{i} \max\set{0, \ell_i(t)-\avg}$ is called \emph{the number of overloaded balls}.
If we enumerate the balls in each bin arbitrarily using natural numbers, this is simply the number of balls whose number is greater than $\avg$.
For instance, in \cref{fig:reordering} (left), the number of overloaded balls is 6.
Note that this is also the number of \enquote{holes} (i.e., $\sum_{i} \max\set{0, \avg-\ell_i(t)}$).

We split this phase into two subphases.
First, we show that it takes $\LDAUOmicron{{(\ln n)}^2/\avg}$ time to reduce the number of overloaded balls to $n$ (\cref{lem:2}).
Afterward (\cref{lem:1}), we prove that if both the discrepancy is logarithmic and the number of overloaded balls is small, a 1-balanced configuration is reached in time $\LDAUOmicron{n/\avg}$.
Together, these immediately imply \cref{lem:potential}.
\begin{lemma}\label{lem:2}
Suppose $\disc(\bm{\ell}(0))=\LDAUOmicron{\ln n}$ and let $T\coloneqq\inf\set{t > 0 \big| \sum_{i} \max\set{0,\ell_i(t) - \avg} \leq n}$.
Then $\Ex{T}=\LDAUOmicron{{(\ln n)}^2/\avg}$.
\end{lemma}
\begin{proof}
Fix a time $t > 0$ and let $A$ denote the number of overloaded balls.
Let $h$ and $k$ be the number of bins with load $>\avg$ and $<\avg$, respectively.
Observe that $\avg-\LDAUOmicron{\ln n}\leq \ell_{\min}\leq \ell_{\max}\leq\avg+\LDAUOmicron{\ln n}$ implies $\min\set{h,k}=\LDAUOmega{A/\ln n}$ and, thus, $h\cdot k=\LDAUOmega{A^2/{(\ln n)}^2}$.

We wait for some ball in some overloaded bin to choose an underloaded bin and move there.
Using \cref{lem:coupling}, we ignore any other move.
There are $h$ overloaded bins and at least $h\cdot\avg$ balls in them.
The probability that such a ball, when activated, chooses an underloaded bin is $k/n$.
Hence, the expected time for such a move to happen is at most
$\frac{1}{h\cdot\avg}\cdot\frac{n}{k}=\LDAUOmicron[small]{\frac{n\cdot{(\ln n)}^2}{A^2\cdot\avg}}$.
When such a move happens,
the value of $A$ decreases by 1.
Hence, total expected time to reduce $A$ to $n$ is bounded by
\begin{align*}
      \sum_{A=n}^{\infty}\LDAUOmicron{\frac{n\cdot{(\ln n)}^2}{A^2\cdot\avg}}
&=      \LDAUOmicron{\frac{n\cdot{(\ln n)}^2}{\avg}\sum_{A=n}^{\infty}A^{-2}}\\
&= \LDAUOmicron{\frac{n\cdot{(\ln n)}^2}{\avg}\int_{n-1}^{\infty}x^{-2}\dif{x}}
=      \LDAUOmicron{{(\ln n)}^2/\avg}
.\qedhere
\end{align*}
\end{proof}

\begin{lemma}\label{lem:1}
Assume that $\disc(\bm{\ell(0)})=\LDAUOmicron{\ln n}$ and that the number of overloaded balls is at most $n$.
Let $T\coloneqq\inf\set{t|\disc(\bm{\ell}(t))\leq1}$.
Then $\Ex{T}=\LDAUOmicron{n/\avg}$.
\end{lemma}
\begin{proof}
Let $A$ denote the number of overloaded balls.
Suppose that $h$ bins have load $>\avg$, $r$ bins have load $=\avg$, and $k$ bins have load $<\avg$.
Note that $h+r+k=n$.
We use the quantity $3A-k-h$ as a potential function and prove the following claim.

\noindent\textbf{Claim.}
if $A>\min\set{h,k}$, then the expected time to decrease $3A-k-h$ by at least 1 is $\leq 3/\avg$.

To see that this implies the lemma, note that we always have $A\geq\max\set{h,k}$;
moreover, if $A=\min\set{h,k}$ then $\avg-1\leq \ell_{\min}\leq \ell_{\max}\leq\avg+1$,
which means the discrepancy is 1.
Since $3A-k-h$ is always between 0 and $3n$ and never increases over time, the claim implies that it takes expected time $\LDAUOmicron{n/\avg}$ to achieve discrepancy $1$.

We prove the claim by considering three cases.
\begin{description}[wide,labelindent=0pt]
\item[Case 1:] $r\geq n/3$ and $A>h$.
    We wait for some ball in a bin with load $>\avg+1$ to choose a bin with load $\avg$ and move there (ignoring any other move via \cref{lem:coupling}).
    Since $A>h$, there is at least one bin with load $>\avg+1$ and, hence, $\geq\avg$ such balls.
    The probability that such a ball, when activated, chooses a bin with load $\avg$ is $r/n$.
    Hence the expected time for such a move to happen is at most $\frac{1}{\avg}\cdot\frac{n}{r}\leq3/\avg$.
    When such a move happens, $A$ and $k$ do not change but $h$ increases by 1, so the potential decreases by 1, as required.
\item[Case 2:] $r\geq n/3$ and $A>k$.
    We wait for some ball in a bin with load $\avg$ to choose a bin with load $<\avg-1$ and move there (ignoring any other move via \cref{lem:coupling}).
    Since $A>k$, there is at least one bin with load $<\avg-1$ and, hence, the expected time for such a move to happen is at most $\frac{1}{r\cdot\avg}\cdot\frac{n}{1}\leq3/\avg$.
    When such a move happens, $A$ and $h$ do not change but $k$ increases by 1, so the potential decreases by 1, as required.
\item[Case 3:] $r<n/3$.
    Note that $h+r+k=n$ and, since we are not yet balanced, $h,k\geq1$.
    So $h+k>2n/3$ gives $h\cdot k\geq h+k-1>2n/3-1>n/3$.
    We wait for some ball in an overloaded bin to choose an underloaded bin and move there (ignoring any other move via \cref{lem:coupling}).
    There are $h$ overloaded bins and, hence, $>h\cdot\avg$ such balls.
    The probability that such a ball, when activated, choose an underloaded bin is $k/n$.
    Hence the expected time for such a move to happen is at most $\frac{1}{h\cdot\avg}\cdot\frac{n}{k}\leq3/\avg$.
    When such a move happens, the value of $A$ decreases by 1, while the values of $k,h$ can decrease by at most 1.
    So, the potential decreases by at least 1, as required.
\qedhere
\end{description}
\end{proof}

\subsection{Phase 3: Reaching Perfect Balance}\label{sec:analysis:homogeneous:phase3}
\begin{lemma}\label{lem:last}
Consider an initial configuration $\bm{\ell}=\bm{\ell}(0)$ with $\disc(\bm{\ell})\leq1$.
Let $T\coloneqq\inf\set{t|\disc(\bm{\ell}(t))<1}$.
Then $\Ex{T}=\LDAUOmicron{n/\avg}$.
\end{lemma}
\begin{proof}
Note that $\avg-1 \leq \ell_{\min} \leq \ell_{\max}\leq \avg + 1$.
By \cref{lem:coupling}, we may ignore any movement of balls from bins with load exactly $m/n$.
Suppose there are $A$ bins of load $>\avg$, and so there are also $A$ bins of load $<\avg$.
If the configuration is not already balanced, then $A\geq 1$ and there are $>\avg\cdot A$ balls that, when activated, find a bin with load $<\avg$ with probability $ A/n$.
The expected time for the first such move to happen is at most
\(
\frac{1}{A \cdot\avg}\cdot\frac{1}{A/n}=\frac{n}{\avg\cdot A^2}
\).
Since $A$ decreases one by one until balancing out (and at that point we would have $A\geq1$), the expected total time to balance out is at most $\sum_{A=1}^{n}\frac{n}{\avg\cdot A^2}\leq\LDAUOmicron{n/\avg}$, as required.
\end{proof}

\section{Conclusion}\label{sec:conclusion}

We analyzed the randomized local search protocol that was first introduced by
\citet{goldberg} and showed that the protocol achieves perfect balance in
expected time $\LDAUOmicron{\ln(n) + n^2/m}$. Moreover, there is a matching
lower bound so our analysis is tight. We now present a few possible future
directions.

The first direction is to extend the analysis to the setting where the bins may
have different speeds, and the load of a bin is defined as its number of balls
divided by its speed. One can consider a similar protocol to \shortlocalsearch:
a ball chooses a random bin on activation, and moves there if and only if doing
so improves its load. A second direction is to study the protocol when the balls
may have different weights. In particular, can we obtain similar balancing times
in the weighted case as in the non-weighted case?  The third direction is to
analyze the protocol in network topologies other than the complete graph.

\end{document}